\documentclass[letterpaper, 10 pt, conference]{ieeeconf} % Comment this line out if you need a4paper

\IEEEoverridecommandlockouts                              % This command is only needed if 
\usepackage{graphics} % for pdf, bitmapped graphics files
\usepackage{epsfig} % for postscript graphics files
\usepackage{amsmath} % assumes amsmath package installed
\usepackage{amssymb}  % assumes amsmath package installed
\usepackage{xcolor}

\usepackage{booktabs}
\usepackage{multirow}
\usepackage{url}

\newtheorem{cond}{Condition}
\newtheorem{theorem}{Theorem}

\title{\LARGE \bf
Online learning of neural state-space models
}

\author{
Bendegúz Györök,
Tamás Péni,
Maarten Schoukens,
Roland Tóth
\thanks{This work was funded by the Air Force Office of Scientific Research under award number FA8655-23-1-7061, by the MathWorks Inc., and by the European Union (ERC, COMPLETE, 101075836). Views and opinions expressed are however those of the authors only and do not necessarily reflect those of the European Union or the European Research Council Executive Agency. Neither the European Union nor the granting authority can be held responsible for them.}% <-this % stops a space
\thanks{B. Györök, T. Péni, and R. Tóth are with the Systems and Control Laboratory, HUN-REN Institute for Computer Science and Control, 1111 Budapest, Hungary (e-mail: gyorokbende@sztaki.hu; peni@sztaki.hu).}%
\thanks{M. Schoukens and R. Tóth are with the Control Systems Group, Eindhoven University of Technology, 5612 Eindhoven, The Netherlands (e-mail: m.schoukens@tue.nl; r.toth@tue.nl).}%
}

\begin{document}

\maketitle
\thispagestyle{empty}
\pagestyle{empty}

%%%%%%%%%%%%%%%%%%%%%%%%%%%%%%%%%%%%%%%%%%%%%%%%%%%%%%%%%%%%%%%%%%%%%%%%%%%%%%%%
\begin{abstract}
Recent advances in deep-learning-based nonlinear system identification have led to encoder-based estimation of \emph{neural state-space} (ANN-SS) models that achieve state-of-the-art performance in offline settings by estimating initial model states from past input-output data. These methods are typically used in multiple-shooting-based offline identification, and online learning of these models remains largely unexplored. This paper presents a batch-wise learning pipeline and a direct recursive identification algorithm for subspace encoder-based ANN-SS models. We provide convergence analysis of the recursive formulation and validate its performance through extensive simulation studies. The results demonstrate that the proposed approach enables computationally efficient online adaptation with high model accuracy.
\end{abstract}

%%%%%%%%%%%%%%%%%%%%%%%%%%%%%%%%%%%%%%%%%%%%%%%%%%%%%%%%%%%%%%%%%%%%%%%%%%%%%%%%
\section{Introduction}
Continuously rising performance requirements and the growing complexity of engineering systems have resulted in a high demand for accurate nonlinear models. For example, reliable path planning and motion control of autonomous vehicles require accurate personalized dynamic models. In practice, it is attractive to use system identification approaches to obtain accurate models directly from measurement data. Recent data-driven modeling methods applying \emph{artificial neural networks} (ANNs), in particular ANN-SS methods, offer superior model performance over conventional nonlinear identification approaches~\cite{ljung_deep_2020}. These models can be estimated with single-shooting~\cite{bemporad_l-bfgs-b_2025} and multiple-shooting strategies~\cite{forgione_continuous-time_2021}. In the latter setting, state-of-the-art approaches, such as~\cite{beintema_deep_2023}, utilize a so-called encoder network to learn the reconstructability map of the model and estimate the initial states from past \emph{input-output} (IO) data.

While offline learning of ANN-SS models often yields accurate and reliable identification of the underlying system, their predictive performance may degrade over time. Variations in operating conditions due to, e.g., component wear, changing loads, or environmental influences can introduce discrepancies between the identified model and the actual system. Furthermore, in many practical applications, system identification is performed using data from a single representative system, while the resulting model is expected to generalize across an entire class of related systems that might significantly vary due to manufacturing tolerances. In addition, acquiring sufficiently informative datasets through dedicated experimental campaigns can be infeasible in practice. In these scenarios, online learning of system dynamics is particularly attractive, as it enables model adaptation using measurement data collected during system operation.

Despite the success of ANN-SS models, online learning and recursive formulations for these architectures remain largely underexplored. Available recursive state-space estimation techniques mostly apply linear model structures~\cite{cho_fast_1994} or linear-in-the-parameters formulations~\cite{wigren_recursive_2023}. Existing online learning methods for nonlinear parametric models typically rely on extensions of the Kalman filter, such as the \emph{Extended Kalman Filter} (EKF); see, e.g., \cite{frascati_online_2026}. However, due to its large computational burden, EKF-based approaches are typically limited to static models and are not feasible for large-scale ANN parametrizations. To estimate ANN-SS models, in \cite{forgione_adaptation_2023}, a Jacobian-based parameter adaptation scheme is combined with a state initialization method that uses past IO data. While conceptually similar to encoder-based methods, this approach assumes that the predicted output is included in the model state and does not provide an explicit characterization of the learned state reconstruction map, such as an encoder network. More recently, meta-learning structures developed for system identification have shown strong adaptation capabilities~\cite{chakrabarty_meta-learning_2023,piga_adaptation_2024}. However, the substantial training overhead associated with meta-models limits their applicability in scenarios requiring rapid adaptation to changing conditions.

Motivated by these limitations, this paper extends classical recursive identification concepts~\cite{ljung_theory_1987} for ANN-SS models, utilizing an encoder network for state estimation. The resulting approach is similar in terms of model structure to the offline SUBNET method~\cite{beintema_deep_2023}, which uses a state-of-the-art encoder-based ANN-SS model. To enable efficient recursive learning in the nonlinear setting, we build upon the recently proposed identification pipeline of~\cite{bemporad_l-bfgs-b_2025}, leveraging \texttt{JAX}, a highly efficient \emph{automatic differentiation} (AD) library. Implementation of the proposed method is available online\footnote{https://github.com/AIMotionLab-SZTAKI/online-sysid}.

The main contributions of this paper are:
\begin{enumerate}
    \item An online recursive identification scheme for ANN-SS models utilizing an encoder for initial state estimation.
    \item Convergence proof of the estimated parameters using the proposed recursive scheme.
    \item An alternative batch-wise online learning pipeline that combines a \texttt{JAX}-based identification approach with the encoder-based ANN-SS model formulation.
\end{enumerate}

\section{Problem setting}\label{sec:problem_statement}

\subsection{System description}
Consider that the dynamics of the data-generating system are represented by the following \emph{discrete-time} (DT) nonlinear \emph{state-space} (SS) form:
\begin{equation}
\label{eq:data-gen-sys}
    x_{k+1}=f\left(x_k, u_k\right),\quad
    y_k=h\left(x_k, u_k\right) + e_k,
\end{equation}
where $k \in \mathbb{Z}$ is the discrete time index, $x_k\in\mathbb{R}^{n_\mathrm{x}}$ is the state, $u_k\in\mathbb{R}^{n_\mathrm{u}}$ is the control input, $y_k\in\mathbb{R}^{n_\mathrm{y}}$ is the measured output, $f:\mathbb{R}^{n_\mathrm{x}} \times \mathbb{R}^{n_\mathrm{u}} \rightarrow \mathbb{R}^{n_\mathrm{x}}$ and $h:\mathbb{R}^{n_\mathrm{x}}\times \mathbb{R}^{n_\mathrm{u}} \rightarrow \mathbb{R}^{n_\mathrm{y}}$ are possibly nonlinear functions, and $e_k$ represents the measurement noise by an i.i.d. white noise process with finite variance. We apply an \emph{output error}~(OE) noise setting for simplicity, but note that more sophisticated noise models can also be considered.

\subsection{Data collection}
We consider that measured \emph{input-output} (IO) data arrives in batches of fixed length $N\geq 1$ with no overlap between the batches. The $i$\textsuperscript{th} batch arriving at time instance $k_{\mathrm{e},i}$ is defined as $\mathcal{B}_i = \{(u_k, y_k)\}_{k=k_{\mathrm{s},i}}^{k_{\mathrm{e},i}}$ for an arbitrary start index $k_{\mathrm{s},i}\in\mathbb{Z}$ such that $k_{\mathrm{e},i}=k_{\mathrm{s},i}+N-1$. The complete dataset accumulated up to batch $i$ is $\mathcal{D}_i = \mathcal{B}_i \cup \mathcal{B}_{i-1} \cup \dots \cup \mathcal{B}_0$. In practice, retraining on the full $\mathcal{D}_i$ is computationally expensive. We therefore introduce a \emph{replay buffer} $\mathcal{M}_i$ of fixed capacity $m \in \mathbb{Z}_{>0}$, defined as the sliding window of the $m$ most recent batches. More formally, $\mathcal{M}_i = \mathcal{D}_i$ for all $i < m$, i.e., the buffer coincides with the full dataset until it reaches capacity. For $i \geq m$, batch $\mathcal{B}_{i-m}$ is deleted and the buffer size is fixed at $mN$ samples. Hence, unlike in offline identification, the full dataset is never available. Note that this setting also involves the special cases when only the recent batch is available, i.e., $m=1$; furthermore, when only one sample arrives at a time, i.e., $N=1$. In the latter scenario, the predictor maps the past IO data to the current model output. To also include a state propagation step, $N\geq 2$ is needed.

\subsection{Learning problem}
Using the available data set $\mathcal{M}_i$, we aim to identify models that characterize \eqref{eq:data-gen-sys} in the form of
\begin{subequations}
\label{eq:ANN-SS}
\begin{equation}\label{eq:ANN_fh}
    \hat{x}_{k+1}=f_\theta(\hat{x}_k, u_k),\quad
    \hat{y}_k=h_\theta(\hat{x}_k, u_k),
\end{equation}
where $\hat{x}_k\in\mathbb{R}^{\hat{n}_\mathrm{x}}$ is the model state, $\hat{y}_k\in\mathbb{R}^{n_\mathrm{y}}$ is the model output, while $f_\theta : \mathbb{R}^{\hat{n}_\mathrm{x}} \times \mathbb{R}^{n_\mathrm{u}}\rightarrow \mathbb{R}^{\hat{n}_\mathrm{x}}$ and $h_\theta : \mathbb{R}^{\hat{n}_\mathrm{x}} \times \mathbb{R}^{n_\mathrm{u}}\rightarrow \mathbb{R}^{n_\mathrm{y}}$ are the state transition function and output map of the model, respectively. Both of these functions are parametrized as simple, feedforward neural networks. Under these conditions, \eqref{eq:ANN_fh} is referred to as a \emph{neural state-space} model. We also consider that an encoder network, capable of estimating the initial states from past IO data, is simultaneously learned with \eqref{eq:ANN_fh}, in the form of
\begin{equation}
    \hat{x}_k = \Psi_\theta(u_{k-n}^{k-1}, y_{k-n}^{k-1}),\label{eq:ANN_enc}
\end{equation}
\end{subequations}
where $u_{k-n}^{k-1}=[u_{k-n}^\top \cdots u_{k-1}^\top]^\top$ and $y_{k-n}^{k-1}$ is defined similarly, while $n\in\mathbb{Z}_+$ is the encoder lag, i.e., the length of the state initialization window. Existence of this encoder is proven based on a state reconstructability argument; see \cite{beintema_deep_2023} for the full derivation. Similarly to $f_\theta$ and $h_\theta$, $\Psi_\theta$ is also parametrized as a feedforward ANN with $\theta\in\mathbb{R}^{n_\theta}$ collecting the parameters of all three networks.

The objective of this work is to develop an online learning method that estimates and continuously updates the model parameters of \eqref{eq:ANN-SS} based on available data sequences $\mathcal{M}_i$.

\section{Batch-wise model learning}\label{sec:continual_learning}

\subsection{Objective function}
Offline ANN-SS identification methods, under the output error noise structure, typically estimate $\theta$ by minimizing the simulation loss. In the current setting, this could be achieved by forward simulating~\eqref{eq:ANN-SS} on the concatenated sequences from the replay buffer $\mathcal{M}_i$. However, this becomes computationally demanding for $mN\gg 1$ and limits opportunities for parallelization. Instead, we evaluate the model independently on each batch. With this modification, the resulting objective function resembles the truncated simulation error loss, see, e.g., \cite{forgione_continuous-time_2021}, and improves computational efficiency while yielding a more favorable optimization landscape~\cite{beintema_deep_2023}. An important aspect of this formulation is that the initial state corresponding to each batch in the replay buffer has to be estimated. One approach is to fix all initial states, e.g., to zero. For a fading-memory system, the resulting mismatch becomes negligible after a sufficiently long transient horizon $k\geq k_0$, but discarding the corresponding samples is undesirable in the online setting. Alternatively, the initial states can be optimized jointly with $\theta$, at the cost of introducing $m\hat{n}_\mathrm{x}$ additional optimization variables. To address these issues, we jointly learn an encoder network together with the model parameters, to estimate the initial states from the last $n$ IO samples. In the offline setting, an ablation study in \cite{beintema_deep_2023} shows that the encoder-based formulation converges the fastest among these approaches. The resulting cost function is
\begin{subequations}\label{eq:cost_fun}
\begin{align}
    V_{\mathcal{M}_i}(\theta) &= \frac{1}{\vert\mathcal{M}_i\vert}\sum_{j\in\mathcal{I}_{i}} \sum_{k=j}^{j+N-1} \|y_{k} - \hat{y}_{k\vert j}\|_2^2,\\
        \hat{x}_{k+1\vert j} &= f_\theta(\hat{x}_{k\vert j}, u_{k}),\quad
    \hat{y}_{k\vert j} = h_\theta(\hat{x}_{k\vert j}, u_{k}),\\
    \hat{x}_{j\vert j} &= \Psi_\theta(u_{j-n}^{j-1}, y_{j-n}^{j-1}),
\end{align}
\end{subequations}
where $\vert\mathcal{M}_i\vert=\min\{i+1, m\}N$, while  $\mathcal{I}_{i}$ collects the starting times of the batches in the replay buffer, i.e., $j$ represents $k_{\mathrm{s},i}$. Furthermore, the pipe notation ($\vert$) is introduced such that $\hat{x}_{k\vert j}$ and $\hat{y}_{k\vert j}$ (with $k\geq j$) denote the simulated state and output at time $k$ from the initial state $\hat{x}_{j\vert j}$. Under model structure \eqref{eq:ANN-SS} and objective function \eqref{eq:cost_fun}, the proposed formulation can be viewed as the SUBNET method (see \cite{beintema_deep_2023}) modified for online learning. A key difference is that the subsections now coincide with the arriving batches, so the number of simulated trajectories is bounded by $m$. The finite buffer acts as a forgetting window: older batches are discarded, which enables tracking of time-varying dynamics. This effect can be more aggressively enforced by including a forgetting factor in \eqref{eq:cost_fun} to assign lower weights to older samples, as commonly applied in recursive identification (see~\cite{ljung_system_1998}). Lastly, although the current paper considers only ANN-SS models, the proposed learning pipeline can be extended to the entire SUBNET-based model family, e.g., for \emph{linear parameter-varying}~(LPV) models~\cite{verhoek_deep-learning-based_2022}, Koopman models~\cite{iacob_deep_2021}, model augmentation~\cite{hoekstra_learning-based_2026}, etc.

\subsection{Parameter estimation}
The model estimate $\hat{\theta}_i$ is obtained based on the dataset $\mathcal{M}_i$ by minimizing \eqref{eq:cost_fun}. The data collection and training processes are illustrated in Fig.~\ref{fig:data_collection}. Unlike offline training, where the optimization is run till convergence from a given (often random) initialization, in the considered setting, in each training phase, the model parameters are initialized with the previous estimate $\hat{\theta}_{i-1}$ and iterations are run for a fixed number of steps. The selection of an efficient optimizer is thus even more critical than in the offline setting. Instead of relying on the original SUBNET implementation~\cite{beintema_deep_2023}, we apply highly efficient automatic differentiation using \texttt{JAX}, similar to~\cite{bemporad_l-bfgs-b_2025}. We perform a fixed number of Adam steps \cite{kingma_adam_2015} followed by L-BFGS~\cite{liu_limited_1989} iterations. 
Adam, a gradient descent variant, is a computationally efficient algorithm used for warm-starting the L-BFGS method. The L-BFGS optimizer, as a quasi-Newton approach, is computationally more demanding, but offers rapid convergence compared to gradient descent-based methods. Careful combination of these optimizers has recently emerged as a promising direction in offline identification, enabling rapid model learning across a variety of problems~\cite{bemporad_l-bfgs-b_2025,bemporad_efficient_2025,gyorok_data-driven_2026}.

\begin{figure}
    \centering
    \includegraphics[width=0.95\linewidth]{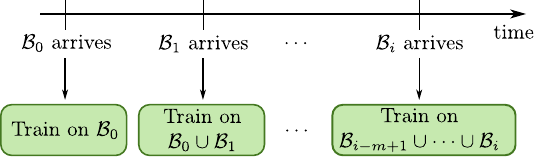}
    \caption{Illustration of the training phases for batch-wise learning.}\vspace{-18pt}
    \label{fig:data_collection}
\end{figure}

Compared to~\cite{bemporad_l-bfgs-b_2025}, several implementation choices were introduced to meet the runtime requirements of online learning. First, in~\cite{bemporad_l-bfgs-b_2025}, the L-BFGS-B optimizer is used; instead, we opted for the L-BFGS algorithm, as the absence of box constraints enables even further reduction of optimization time. Second, before the replay buffer is fully populated ($i<m$), unused entries are filled with zeros and excluded from the cost function via masking. This keeps the optimization variables and computational graph fixed throughout operation, enabling efficient batched simulations using \texttt{JAX} vectorization. Finally, the entire learning pipeline, including model evaluation, gradient computation, and optimization, is compiled offline, using \emph{just-in-time} (JIT) compilation before the first data batch arrives. Since the computational graph remains unchanged, no recompilation is required during operation, allowing model updates to be performed with minimal overhead. These choices ensure the efficiency required for real-time deployment.

Another characteristic of the formulation is that, under mild assumptions, it inherits convergence and consistency properties of the SUBNET method as $m\to\infty$, i.e., when all batches are retained. In practice, storing all previously collected data is infeasible. Hence, a finite $m$ trades the asymptotic properties for real-time performance: a trade-off absent in the offline setting. The recursive approach, introduced subsequently, addresses this limitation without requiring infinite data storage.

\section{Recursive parameter updates}\label{sec:recursive_learning}

\subsection{Recursive estimator}
While the batch-wise learning of Sec.~\ref{sec:continual_learning} is effective for moderate-sized buffers, it remains demanding for low-latency online applications. We therefore derive a recursive scheme that provides a computationally cheaper alternative by updating parameters in a single step upon each batch arrival. Unlike the batch-wise approach, it stores only the current batch, i.e., $m=1$, with memory represented by the parameter estimates $\theta_i$. With each incoming batch, parameters are updated using the recursive Gauss–Newton rule:
\begin{subequations}\label{eqs:Gauss_Newton_alg}
\begin{align}
    \theta_{i} &= \theta_{i-1} + \mu_i\, R_i^{-1}\, \psi_i^\top(\theta_{i-1})\,\epsilon_i(\theta_{i-1}),\label{eq:recursive_param_update}\\
    R_i &= (1 - \mu_i) R_{i-1} + \mu_i\, \psi_i^\top(\theta_{i-1})\psi_i(\theta_{i-1}),\label{eq:recursive_hessian_update}
\end{align}
\end{subequations}
where $\theta_i$ is the estimate of $\theta$ after batch $i$, $\mu_i$ is the step size, $\psi_i(\theta_{i-1})=\frac{\partial \hat{Y}_i}{\partial\theta}\vert_{\theta_{i-1}}$ and $\epsilon_i(\theta_{i-1})=Y_i-\hat{Y}_i(\theta_{i-1})$ are the Jacobian and prediction error for batch $\mathcal{B}_i$. Details on the efficient computation of \eqref{eq:recursive_param_update} and on step-size selection are given in Sec.~\ref{sec:practical_guide}. Denoting the starting time of batch $\mathcal{B}_i$ by $k_{\mathrm{s},i}$, the batched output and predictor vectors are
\begin{subequations}
\begin{align}
    Y_i &= [ y_{k_{\mathrm{s},i}}^\top\ y_{k_{\mathrm{s},i}+1}^\top\ \cdots\ y_{k_{\mathrm{s},i}+N-1}^\top ]^\top,\\
    \hat Y_i(\theta) &= [h_\theta^\top(\hat x_{k_{\mathrm{s},i}}, u_{k_{\mathrm{s},i}})\, \cdots\, h_\theta^\top(\hat{x}_{k_{\mathrm{s},i}+N-1}, u_{k_{\mathrm{s},i}+N-1})]^\top,\\
    \hat{x}_{k_{\mathrm{s},i}} &= \Psi_\theta(u_{k_{\mathrm{s},i}-n}^{k_{\mathrm{s},i}-1}, y_{k_{\mathrm{s},i}-n}^{k_{\mathrm{s},i}-1}),
\end{align}
\end{subequations}
where $\hat{x}_{k_{\mathrm{s},i}}$ is computed by the encoder using past data, and the remaining states are obtained by forward propagating $f_\theta$.

Algorithm~\eqref{eqs:Gauss_Newton_alg} is related to the conventional \emph{recursive prediction-error minimization} (RPEM) method of~\cite{ljung_analysis_1981}, but differs in the form of recursion: conventional RPEM updates the parameters after each sample, whereas here, recursive updates are performed for batches of $N$ samples. Thus, the prediction error $\epsilon_i$ and Jacobian $\psi_i$ are formulated as batched quantities. The proposed method also supports a batch size of $N=1$, recovering the conventional one-sample adaptation.

Furthermore, \eqref{eqs:Gauss_Newton_alg} minimizes the running average criterion
\begin{equation}\label{eq:running_cost}
    V_{T}(\theta) = \frac{1}{T}\sum_{i=0}^{T-1}\|\epsilon_i(\theta)\|_2^2,
\end{equation}
where $T$ is the number of batches observed so far. Introducing a forgetting factor (see Sec.~\ref{sec:practical_guide}) replaces the uniform weighting of batches in \eqref{eq:running_cost} with an exponentially decaying one. Similar to the batch-wise learning approach in Sec.~\ref{sec:continual_learning}, we assume that at least $n$ extra samples are available before the first batch to enable encoder-based state estimation.

\subsection{Gradient computation}
Gradient computation is particularly important for ANN-SS models, as $\hat{Y}_i$ depends on the parameters through the nonlinear state dynamics, output mapping, and the encoder-based initial state, unlike the simpler parametric models commonly considered in the RPEM setting. Hence, the required Jacobian $\psi_i(\theta_{i-1}) \in\mathbb{R}^{Nn_\mathrm{y}\times n_\theta}$ is obtained by propagating state sensitivities $s_k=\frac{\partial \hat{x}_k}{\partial\theta}\vert_{\theta=\theta_{i-1}} \in\mathbb{R}^{\hat{n}_\mathrm{x}\times n_\theta}$, following~\cite{forgione_adaptation_2023}. Let the index set $\mathbb{I}_i = \{k_{\mathrm{s},i},\dots, k_{\mathrm{s},i}+N-1\}$ denote the time indices associated with batch $\mathcal{B}_i$. With the encoder providing the initial sensitivity, the $s_k$ values can be expressed as
\begin{equation}
    s_{k_{\mathrm{s},i}} = \frac{\partial \Psi_\theta(u_{k_{\mathrm{s},i}-n}^{k_{\mathrm{s},i}-1}, y_{k_{\mathrm{s},i}-n}^{k_{\mathrm{s},i}-1})}{\partial\theta}\Bigr\vert_{\theta_{i-1}},\,\,
    s_{k+1} = F_{\mathrm{x},k}\, s_k + F_{\theta,k},
\end{equation}
where $k\in\mathbb{I}_i$, and $F_{\mathrm{x},k}\in\mathbb{R}^{\hat{n}_\mathrm{x}\times \hat{n}_\mathrm{x}},F_{\theta,k}\in\mathbb{R}^{\hat{n}_\mathrm{x}\times n_\theta}$ are the Jacobians of $f_\theta(\hat{x}_k, u_k)$ w.r.t. $\hat{x}_k$ and $\theta$, evaluated at $\theta_{i-1}$. The rows of $\psi_i(\theta_{i-1})$ are then
\begin{equation}
    \frac{\partial \hat y_k}{\partial \theta} = H_{\mathrm{x},k}\, s_k + H_{\theta,k},\quad k \in\mathbb{I}_i,
\end{equation}
with $H_{\mathrm{x},k}\in\mathbb{R}^{n_\mathrm{y}\times \hat{n}_\mathrm{x}}$ and $H_{\theta,k}\in\mathbb{R}^{n_\mathrm{y}\times n_\theta}$ being the Jacobians of $h_\theta(\hat{x}_k, u_k)$ w.r.t. $\hat{x}_k$ and $\theta$, respectively, evaluated at $\theta_{i-1}$. In practice, the required Jacobians are computed via automatic differentiation using the \texttt{JAX} library, similar to Sec.~\ref{sec:continual_learning}.

\subsection{Convergence analysis}
Recursion \eqref{eqs:Gauss_Newton_alg} is an instance of the general stochastic-approximation scheme analyzed in \cite{ljung_analysis_1981,ljung_theory_1987,ljung_system_1998}. However, classical results consider only one-sample adaptation without an encoder network in the model. Hence, we extend the results of~\cite[Ch. 4]{ljung_theory_1987} to establish that $\theta_i$ converges almost surely to a set of stationary points. 

Let us denote the batch predictor of the model structure \eqref{eq:ANN-SS}, associated with batch $\mathcal{B}_i$ as $\hat\gamma$, expressed as
\begin{equation}
    \hat Y_i(\theta) = \hat{\gamma}(\theta,\, u_{k_{\mathrm{s},i}-n}^{k_{\mathrm{s},i}+N-1}, y_{k_{\mathrm{s},i}-n}^{k_{\mathrm{s},i}-1}).
\end{equation}
Note that the predictor uses measured outputs only over the state initialization window $k\in [k_{\mathrm{s},i}-n, k_{\mathrm{s},i}-1]$. Over the simulation window $k\in [k_{\mathrm{s},i}, k_{\mathrm{s},i}+N-1]$, the predictor only uses the available input signals. We make the general assumption that $\hat{\gamma}$ is differentiable w.r.t. $\theta$ everywhere on $\Theta$. As parameter estimation is implemented using automatic differentiation in practice, this is only a technical condition. 

To prove convergence of the recursive scheme, we require that system~\eqref{eq:data-gen-sys} is incrementally exponentially output stable.
\begin{cond}\label{cond:stab_data_gen}
    For any $\delta>0$, there exist a $C(\delta)\in[0,\infty)$ and a $\lambda\in[0,1)$, such that
    \begin{equation}
        \mathbb{E}_e\{\|y_k - \tilde{y}_k\|_2^4\} < C(\delta) \lambda^{k-k_0},\quad \forall k\geq k_0
    \end{equation}
    for all $k_0\geq 0$ and all $x_{k_0},\tilde{x}_{k_0}\in\mathbb{R}^{n_\mathrm{x}}$ with $\|x_{k_0} - \tilde{x}_{k_0}\|_2 < \delta$, and $\{(u_k, e_k)\}_{k=k_0}^{\infty}\in\mathcal{W}_{[k_0,\infty]}$, where $\mathcal{W}_{[k_0,\infty]}$ denotes the $\sigma$-algebra associated with the random variables $\{(u_k, e_k)\}_{k=k_0}^{\infty}$; moreover, $y_k$ and $\tilde{y}_k$ satisfy \eqref{eq:data-gen-sys} with the same $(u_k, e_k)$, but with different initial conditions.
\end{cond}

Note that Condition~\ref{cond:stab_data_gen} is conventional for classical convergence proofs (and even consistency analysis) of parametric identification approaches; see, e.g., \cite{ljung_system_1998}. For our analysis, we also require that the applied model satisfies the following stability condition w.r.t. perturbations of the dataset:
\begin{cond}\label{cond:stab_pred}
    There exists $C\in[0,\infty)$, such that for all $\theta\in\Theta$ and $k_0\geq n$, the predictor $\hat{\gamma}$ satisfies
    \begin{multline}\label{eq:model_stab}
        \|\hat{\gamma}(\theta, u_{k_0-n}^{k_0+N-1}, y_{k_0-n}^{k_0-1}) - \hat{\gamma}(\theta, \tilde{u}_{k_0-n}^{k_0+N-1}, \tilde{y}_{k_0-n}^{k_0-1})\|_2 \\\leq C \sum_{\ell=k_0-n}^{k_0+N-1}\|u_\ell - \tilde{u}_\ell\|_2 + C\sum_{\ell=k_0-n}^{k_0-1} \|y_\ell - \tilde{y}_\ell\|_2,
    \end{multline}
    for any $(u_{k_0-n}^{k_0+N-1}, y_{k_0-n}^{k_0-1})$ and $(\tilde{u}_{k_0-n}^{k_0+N-1}, \tilde{y}_{k_0-n}^{k_0-1})$. Further, $\|\hat{\gamma}(\theta, 0_{k_0-n}^{k_0+N-1}, 0_{k_0-n}^{k_0-1})\|_2\leq C$, and \eqref{eq:model_stab} is satisfied by $\frac{\partial}{\partial\theta}\hat{\gamma}$.
\end{cond}

An advantage of using an encoder is that it resets the state sensitivity $s_{k_{\mathrm{s},i}}$ for each batch, so $\psi_i$ is independent of previous model evaluations. Thus, the classical global incremental predictor stability condition of~\cite{ljung_analysis_1981} reduces to the finite-horizon Lipschitz bound in Condition~\ref{cond:stab_pred}. Moreover, \eqref{eq:model_stab} holds if $\Psi_\theta$, $f_\theta$, and $h_\theta$ are uniformly Lipschitz over the compact set $\Theta$, which is satisfied by the applied ANN parametrization with Lipschitz activation functions (e.g., tanh, relu). The requirement on $\frac{\partial}{\partial\theta}\hat{\gamma}$ is also satisfied for activation functions with bounded Lipschitz derivatives (e.g., tanh) on bounded IO sets.

To prove convergence of \eqref{eqs:Gauss_Newton_alg}, we first show that in the limit $T\to\infty$, the cost function $V_T(\theta)$ converges to a deterministic criterion, based on classical results~\cite{ljung_convergence_1978}.

\begin{theorem}\label{thm:cost_fun_convergence}
    Let Conditions~\ref{cond:stab_data_gen}--\ref{cond:stab_pred} hold with a quasi-stationary input process $u$ independent of the white-noise process $e$. Then,
    \begin{equation}\label{eq:cost_fun_convergence}
        \sup_{\theta\in\Theta} \vert V_T(\theta) - \bar{V}(\theta)\vert \to 0,
    \end{equation}
    with probability 1 as $T\to\infty$, where $\bar{V}(\theta) = \lim_{T\to\infty}\frac{1}{T}\sum_{i=0}^{T-1} \mathbb{E}\{\|\epsilon_i(\theta)\|_2^2\}$. With $\psi_i(\theta)=\frac{\partial \hat{Y}_i}{\partial\theta}\vert_\theta$, $\bar{V}$ is continuously differentiable, and the following limits exist and are continuous on $\Theta$:
    \begin{align}
        g(\theta)&:=\lim_{T\to\infty}\frac{1}{T} \sum_{i=0}^{T-1}\mathbb{E}\{\psi_i^\top(\theta)\epsilon_i(\theta)\},\label{eq:g_theta}\\
        \bar {R}(\theta) &:= \lim_{T\to\infty} \frac{1}{T}\sum_{i=0}^{T-1}\mathbb{E}\{\psi_i^\top(\theta)\psi_i(\theta)\},\label{eq:Rbar_theta}
    \end{align}
    giving $\tfrac{\mathrm{d}}{\mathrm{d}\theta}\bar{V}(\theta)=-2g(\theta)$.
\end{theorem}
\begin{proof}
    The batch loss $\ell_i(\theta)=\|\epsilon_i(\theta)\|_2^2$ satisfies C1 in \cite{ljung_convergence_1978} and Cond.~\ref{cond:stab_pred} makes $\ell_i$ uniformly Lipschitz in $\theta$ on the compact set $\Theta$. Since the encoder network and the batch formulation ensure that the prediction $\hat{Y}_i$ depends on only a finite window of past IO data, $\hat{\gamma}$ has finite memory, which is a special case of the fading-memory assumption in \cite{ljung_convergence_1978}. Together with the quasi-stationarity of the data and Cond.~\ref{cond:stab_data_gen}, the loss sequence is quasi-stationary with geometrically summable covariances. Consequently, the proof of \cite[Lem. 3.1]{ljung_convergence_1978} applies, giving \eqref{eq:cost_fun_convergence}. Under Cond.~\ref{cond:stab_pred}, \cite[Lem. 3.1]{ljung_convergence_1978} also applies to $\frac{\mathrm{d}}{\mathrm{d}\theta}\ell_i$, yielding \eqref{eq:g_theta} and \eqref{eq:Rbar_theta}, which are both continuous on $\Theta$. Uniform convergence of the gradients together with convergence of $V_T$ then gives $\frac{\mathrm{d}}{\mathrm{d}\theta}\bar{V}(\theta) = -2g(\theta)$.
\end{proof}

Theorem~\ref{thm:cost_fun_convergence} implies that the (local) minima of $V_T(\theta)$ and $\bar{V}(\theta)$ become arbitrarily close as the number of batches increases. We now state the convergence result for $\theta_i$.

\begin{theorem}\label{thm:convergence}
Consider the recursion \eqref{eqs:Gauss_Newton_alg} applied to data generated by \eqref{eq:data-gen-sys}. Assume that the conditions of Theorem~\ref{thm:cost_fun_convergence} hold; together with
\begin{enumerate}
    \item [(a)] step size $\mu_i>0$ satisfies $i\mu_i \to \bar{\mu}\in\mathbb{R}_+$ as $i\to\infty$;
    \item [(b)] $R_i\succ 0$ for all $i$;
    \item [(c)] \eqref{eq:recursive_param_update} includes a projection keeping $\theta_i\in\Theta$.
    \end{enumerate}
    Then, either $\lim_{i\to\infty}\theta_i \in \{\theta \mid\frac{\mathrm{d}}{\mathrm{d}\theta}\bar V(\theta) = 0\}$, 
    or $\theta_i$ converges to the boundary of $\Theta$, with probability 1.
\end{theorem}
\begin{proof}
    Under the considered conditions, the proof of \cite[Thm. 2]{ljung_analysis_1981} applies (also \cite[Ch. 4]{ljung_theory_1987}). The proof uses the ODE associated with recursion \eqref{eqs:Gauss_Newton_alg}. Changing the time-scale from $i$ to $\tau$ such that $\theta_D(\tau_i) \leftrightarrow \theta_i$, $\tau_i=\sum_{\ell=0}^i\mu_\ell$, the iterates $(\theta_i, R_i)$ track the solution of $\frac{\mathrm{d}}{\mathrm{d}\tau}\theta_D(\tau)=R_D^{-1}(\tau)g(\theta_D(\tau))$ with $\frac{\mathrm{d}}{\mathrm{d}\tau}R_D(\tau)=\bar{R}(\theta_D(\tau))-R_D(\tau)$. Then, based on \cite[Thm. 2]{ljung_analysis_1981}, $\bar{V}$ is a Lyapunov function for this ODE system with $\frac{\mathrm{d}}{\mathrm{d}\tau}\bar{V}(\theta_D(\tau))=-2g^\top(\theta_D(\tau)) R_D^{-1}(\tau)g(\theta_D(\tau))\leq 0$. As $R_D^{-1}(\tau)\succ 0$, convergence of \eqref{eqs:Gauss_Newton_alg} follows.
\end{proof}

Theorem~\ref{thm:convergence} shows that \eqref{eqs:Gauss_Newton_alg} is convergent under the specified conditions. Thus, we have constructed a recursive identification scheme for ANN-SS models with encoder-based state initialization that can also be straightforwardly extended to other approaches in the SUBNET model family.

\subsection{Practical guide}\label{sec:practical_guide}
There are a few conditions of Theorem~\ref{thm:convergence} that require further discussion. First, the step size $\mu_i$ needs to satisfy (a). For time-invariant systems, the optimal choice is $\mu_i=1/i$ as it provides minimal variance parameter estimates~\cite{ljung_theory_1987}. However, at the beginning of the recursion, the gradient computations might be rather inaccurate due to the (possible) inaccuracy of the initial parameters. Based on~\cite{ljung_theory_1987}, we therefore introduce a forgetting factor $\lambda_i$, and use the following formulation:
\begin{equation}\label{eq:time_inv_mu}
    \mu_i = \mu_{i-1} / (\mu_{i-1} + \lambda_i),
\end{equation}
where $\mu_0$ is a user-specific constant that influences transient adaptation. As $\lambda_i>0$ ensures forgetting, $\mu_0=1$ is suggested. To satisfy (a) in Theorem~\ref{thm:convergence} and recover the optimal asymptotic law of $\mu_i\sim 1/i$, the forgetting factor should converge to 1 as $i\to\infty$. Hence, a suggested choice is $\lambda_i = \bar{\lambda} \lambda_{i-1} + (1 - \bar{\lambda})$, where $\lambda_0$ and $\bar{\lambda}$ are design parameters, with recommended values of $\bar{\lambda}=0.99$ and $\lambda_0\in [0.75,\, 0.95]$. For time-varying systems, a common choice is to select a constant forgetting factor $\lambda_i\equiv \bar{\lambda}<1$, and, instead of \eqref{eq:time_inv_mu}, apply a step-size, as
\begin{equation}\label{eq:exponential_forgetting_factor}
    \mu_i = (1 - \bar{\lambda})/(1 - \bar{\lambda}^i).
\end{equation}
With $\bar{\lambda}$ close to 1, \eqref{eq:exponential_forgetting_factor} ensures that ''old'' batches contribute marginally to the parameter updates, which is practical for fast-changing dynamics. For a detailed discussion and ablation studies, see \cite[Ch. 5.6]{ljung_theory_1987}.

Another important aspect of Theorem~\ref{thm:convergence} is (b), which requires $R_i$ to be positive definite for all $i\geq0$. This is necessary both for the convergence proof and to compute the inverse $R_i^{-1}$ at each step. Specific to ANN-SS models, $R_i$ can be ill-posed or singular due to the typical over-parametrization. Hence, in practice, $R_i^{-1}$ is computed by $R_i^{-1}\approx (R_i{+}\delta I_{n_\theta})^{-1}$ with a small $\delta\in\mathbb{R}_{>0}$. Naively computing this inverse is computationally costly. Since $R_i$ is symmetric and $\delta I_{n_\theta}$ guarantees positive definiteness, we propose directly computing $R_i^{-1}\psi_i^\top\epsilon_i$ via Cholesky factorization, similar to Gaussian process inference~\cite{rasmussen_gaussian_2005}. Hence, the parameter update in \eqref{eq:recursive_param_update} simplifies to
\begin{subequations}
\begin{align}
    \theta_i &= \theta_{i-1} + \mu_i L_i^\top\backslash(L_i\backslash \psi_i^\top(\theta_{i-1})\epsilon_i^\top(\theta_{i-1})),\\
    L_i &= \mathrm{chol}(R_i + \delta I_{n_\theta}),
\end{align}
\end{subequations}
where $A\backslash b$ denotes the vector $z$ which solves $Az = b$. 
As $R_i$ represents the covariance, i.e., uncertainty, of $\theta_i$, we advise initializing it as a large multiple of the identity matrix, e.g., $R_0=10^4I_{n_\theta}$.

Finally, we note that convergence to the boundary of $\Theta$ is only technical. With particular projection schemes, convergence to local minima of $\bar{V}(\theta)$ can be proven~\cite{ljung_analysis_1981}.

\section{Simulation example}\label{sec:sim_study}
To demonstrate the capabilities of the proposed online learning methods, we identify the nonlinear lateral dynamics of a ground vehicle, expressed as
\begin{subequations}\label{eq:example_data_gen_sys}
\begin{align}
    \dot{v}_\mathrm{y} &= (F_\mathrm{y,f}\cos\delta + F_\mathrm{y,r})/m - v_\mathrm{x}\omega,\\
    \dot{\omega} &= (l_\mathrm{f}F_\mathrm{y,f}\cos\delta - l_\mathrm{r}F_\mathrm{y,r})/I_\mathrm{z},
\end{align}
\end{subequations}
where $v_\mathrm{x}$ is the longitudinal, $v_\mathrm{y}$ is the lateral velocity, $\omega$ is the yaw rate, $\delta$ is the steering angle, and the parameters are given in Table~\ref{tab:phys_params}. The lateral tire forces $F_\mathrm{y,f}$ and $F_\mathrm{y,r}$ are expressed by the Pacejka Magic Formula~\cite[Ch. 4.3]{pacejka_tyre_2012}, as
\begin{multline}
    F_{\mathrm{y},p} = \frac{mgl_p}{l_\mathrm{r} + l_\mathrm{f}} D_p \sin(C_p \arctan(B_p\alpha_p) - \\E_p(B_p\alpha_p - \arctan(B_p\alpha_p))),\quad p\in\{\mathrm{r},\mathrm{f}\},
\end{multline}
with $B_p$, \dots, $E_p$ parameters given in Table~\ref{tab:pacejka}, denoting the Pacejka coefficients, and $\alpha_\mathrm{f}$, $\alpha_\mathrm{r}$ being the slip angles:
\begin{equation}
    \alpha_\mathrm{f} = \delta - \arctan(\frac{v_\mathrm{y}+l_\mathrm{f}\omega}{v_\mathrm{x}}),\quad \alpha_\mathrm{r} = -\arctan(\frac{v_\mathrm{y}-l_\mathrm{r}\omega}{v_\mathrm{x}}).
\end{equation}
Model \eqref{eq:example_data_gen_sys} is discretized with the RK4 scheme under a zero-order hold assumption and a sampling time of $\Delta t = 0.01~\mathrm{s}$. We consider a constant longitudinal velocity throughout the simulations. The state is then $x_k = [v_{\mathrm{y},k}\ \omega_k]^\top\in\mathbb{R}^2$, the control input is $u_k=\delta_k\in\mathbb{R}$, and the measured output is $y_k=\omega_k+e_k$ with $e_k$ being an i.i.d. white noise process representing measurement noise and resulting in 30 dB \emph{signal-to-noise ratio}~(SNR) level. Data is generated by using a multisine input signal with 6 frequency components logarithmically spaced in the range $[0.1,\, 2]$ Hz, with Schroeder phases and an overall amplitude of 0.07 rad.
\begin{table}
  \centering\vspace{4pt}
  \begin{minipage}[c]{0.52\linewidth}
    \centering
    \caption{Physical parameters.}
    \label{tab:phys_params}
    \begin{tabular}{lc}
    \hline
        Parameter & Value\\
        \hline
        Mass $m$ & 1500 kg\\
        Yaw inertia $I_\mathrm{z}$ & $2250~\mathrm{kgm^2}$\\
        Front axle dist. $l_\mathrm{f}$ & 1.2 m\\
        Rear axle dist. $l_\mathrm{r}$ & 1.4 m\\
        Gravity $g$ & 9.81 N/kg\\
        Long. speed $v_\mathrm{x}$ & 20 m/s\\
        \hline
    \end{tabular}
  \end{minipage}\hfill
  \begin{minipage}[c]{0.46\linewidth}
    \centering
    \caption{Pacejka tire parameters.}
    \label{tab:pacejka}
    \begin{tabular}{lcc}
      \hline
       & Front & Rear \\
      \hline
      $B$ & $10.0$ & $11.0$ \\
      $C$ & $1.9$  & $1.9$  \\
      $D$ & $1.0$  & $1.0$  \\
      $E$ & $0.97$ & $0.97$ \\
      \hline
    \end{tabular}
  \end{minipage}\vspace{-8pt}
\end{table}

A batch size of $N=25$ and memory length $m=20$ are used. Training continues until $3\cdot 10^4$ samples (5 minutes of simulation time). The model order is $\hat{n}_\mathrm{x}=2$ with $f_\theta$ parametrized as a feedforward ANN with a linear bypass, 2 hidden layers of 4 nodes, using the \emph{hyperbolic tangent} (tanh) activation function. The encoder is implemented as a feedforward ANN with 1 hidden layer of 16 tanh units. The map $h_\theta$ matches the true output structure, i.e., $\hat{y}_k = [1\ 0]\hat{x}_k$. To illustrate convergence, parameters are saved after each training step and evaluated on separate test data of length $N_\mathrm{test}=600$ generated using the same input distribution.

First, we demonstrate the efficiency of the batch-wise learning approach outlined in Sec.~\ref{sec:continual_learning}. To showcase the importance of the L-BFGS optimizer, we first use only 20 Adam iterations per training phase, then repeat the experiment by appending 2 L-BFGS optimization steps to each phase. Fig.~\ref{fig:example2_convergence} shows convergence results, while Fig.~\ref{fig:example2_times} shows the mean wall time\footnote{Achieved with an Intel Core i7-13700HX CPU and 16 GB of RAM.} of the training phases. As visible, including the extra L-BFGS steps increases the computational load, but it results in rapid convergence compared to when only Adam is used.

\begin{figure}
    \centering
    \includegraphics[scale=0.83]{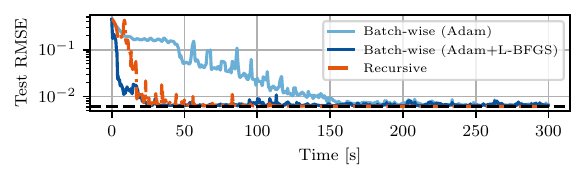}\vspace{-12pt}
    \caption{Test errors after each training phase with three online learning configurations. The dashed black line corresponds to the noise floor.}\vspace{-18pt}
    \label{fig:example2_convergence}
\end{figure}
\begin{figure}
    \centering
    \includegraphics[scale=0.83]{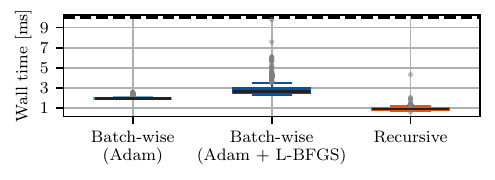}\vspace{-12pt}
    \caption{Distribution of the wall time across the training phases with three methods. The dashed black line corresponds to the sampling time.}\vspace{-12pt}
    \label{fig:example2_times}
\end{figure}

Next, we investigate the recursive scheme outlined in Sec.~\ref{sec:recursive_learning}. The step size is selected as in \eqref{eq:time_inv_mu} with $\bar{\lambda}=0.99$, $\lambda_0=0.75$, $\mu_0=1$. For the matrix inverse computation, $\delta=10^{-3}$ is used, and $R_i$ is initialized as $R_0=10^{4} I_{n_\theta}$. Fig.~\ref{fig:example2_convergence} shows that the convergence speed is slightly slower than with the batch-wise approach, which is expected due to taking fewer iterations in each training phase and applying only a single batch for model updates. However, as Fig.~\ref{fig:example2_times} shows, computations take only about 1 ms per iteration. 

Finally, we present a sensitivity analysis of the proposed learning methods with respect to the batch length $N$ and, in the case of the batch-wise approach, the memory $m$. To compactly characterize convergence properties, we adopt the \emph{area under the convergence curve} (AUCC) metric, obtained by numerically integrating the convergence curves of Fig.~\ref{fig:example2_convergence}. Lower AUCC values indicate faster and more accurate convergence. The results are summarized in Table~\ref{tab:method_comparison}. Three observations emerge. First, the computational cost of the batch-wise method scales with the replay-buffer length $m\cdot N$, while the recursive scheme remains substantially cheaper across all settings. Second, choosing the right batch length is crucial for the batch-wise method: training on short horizons yields inaccurate (and in some cases unstable) models, as seen at $N=5$, where both $m=1$ and $m=100$ scenarios lead to divergence (denoted by $\infty$ AUCC). For sufficiently long horizons, however, the batch-wise method attains the lowest AUCC values overall, outperforming the recursive scheme. Third, the recursive scheme shows robust convergence across all batch lengths. The increase in AUCC as $N$ increases is only due to convergence being measured against the physical time axis; hence, having longer batches inherently increases convergence time. The results show that the choice between the two schemes can be made according to the desired computational and convergence properties. The recursive scheme provides low computational cost and reliable convergence, while the batch-wise method can achieve superior convergence at the expense of higher computational demand and increased sensitivity to hyperparameters.

\begin{table}
\centering
\caption{Avg. computation time and AUCC with varying batch settings.}
\label{tab:method_comparison}
\setlength{\aboverulesep}{0pt}
\setlength{\belowrulesep}{0pt}
\setlength{\cmidrulesep}{0pt}
\begin{tabular}{cc cc cc}
\hline
& \multicolumn{3}{c}{Batch-wise} & \multicolumn{2}{c}{Recursive} \\
\cmidrule(lr){2-4} \cmidrule(lr){5-6}
$N$ & $m$ & Time [ms] & AUCC & Time [ms] & AUCC \\
\hline
\multirow{3}{*}{5}
  & 1 & 0.4 & $\infty$ & \multirow{3}{*}{0.6} & \multirow{3}{*}{3.45} \\
  & 10 & 0.8 & 17.47 & & \\
  & 100 & 2.5 & $\infty$ & & \\
\hline
\multirow{3}{*}{25}
  & 1 & 0.8 & $\infty$ & \multirow{3}{*}{0.9} & \multirow{3}{*}{5.16} \\
  & 10 & 2.0 & 3.16 & & \\
  & 100 & 9.3 & 2.44 & & \\
\hline
\multirow{3}{*}{50}
  & 1 & 1.4 & 46.04 & \multirow{3}{*}{1.4} & \multirow{3}{*}{7.66} \\
  & 10 & 3.7 & 4.62 & & \\
  & 100 & 17.8 & 3.37 & & \\
\hline
\end{tabular}\vspace{-18pt}
\end{table}

\section{Conclusion}\label{sec:conclusions}
This paper investigated online learning of system dynamics using ANN-SS models with an encoder network for initial state estimation. We formulated and analyzed a recursive identification approach extending conventional system identification results. We also adapted a recently proposed, computationally efficient identification pipeline to batch-wise learning. A comprehensive simulation study demonstrated the effectiveness of both methods. Future work will focus on practical aspects, including more complex identification tasks, larger models, and time-varying system dynamics.

\addtolength{\textheight}{-1cm}   % This command serves to balance the column lengths

%%%%%%%%%%%%%%%%%%%%%%%%%%%%%%%%%%%%%%%%%%%%%%%%%%%%%%%%%%%%%%%%%%%%%%%%%%%%%%%%

%%%%%%%%%%%%%%%%%%%%%%%%%%%%%%%%%%%%%%%%%%%%%%%%%%%%%%%%%%%%%%%%%%%%%%%%%%%%%%%%



% \section*{ACKNOWLEDGMENT}

% The preferred spelling of the word ÒacknowledgmentÓ in America is without an ÒeÓ after the ÒgÓ. Avoid the stilted expression, ÒOne of us (R. B. G.) thanks . . .Ó  Instead, try ÒR. B. G. thanksÓ. Put sponsor acknowledgments in the unnumbered footnote on the first page.

\bibliography{IEEEabrv,literature}

%%%%%%%%%%%%%%%%%%%%%%%%%%%%%%%%%%%%%%%%%%%%%%%%%%%%%%%%%%%%%%%%%%%%%%%%%%%%%%%%

\end{document}